\documentclass[manuscript]{BSLstyle} 

\usepackage{mathrsfs}     
\usepackage{bussproofs}	  

\AuthorEmail{Alexander V. Gheorghiu}{alexander.gheorghiu.19@ucl.ac.uk}
\AuthorEmail{David J. Pym}{d.pym@ucl.ac.uk}

\Affiliation{University College London}{Department of Computer Science}{Gower St, London WC1E 6BT}{London, United Kingdom}{1,2}
\Affiliation{University College London}{Department of Philosophy}{Gower St, London WC1E 6BT}{London, United Kingdom}{2}
\Affiliation{University of London}{Institute of Philosophy}{Senate House, Malet St, London WC1E 7HU}{London, United Kingdom}{2}

\Title[Definite formulae, NAF, and the Base-extension Semantics for IPL]{Definite formulae, Negation-as-Failure, and the Base-extension Semantics of Intuitionistic Propositional Logic}

\newcommand{\argument}[1]{\ensuremath{\mathcal{#1}}}
\newcommand{\base}[1]{\ensuremath{\mathscr{#1}}}
\newcommand{\calculus}[1]{\ensuremath{\mathsf{#1}}}

\newcommand{\set}[1]{\ensuremath{\mathbb{#1}}}
\newcommand{\setAtoms}{\set{A}}

\newcommand{\setFormulae}{\set{F}}
\newcommand{\basis}[1]{\mathfrak{#1}}

\newcommand{\sat}{\vDash}
\newcommand{\proves}{\vdash}
\newcommand{\entails}{\vDash}
\newcommand{\supp}{\Vdash}
\newcommand{\seq}{\triangleright}

\newcommand{\at}[1]{{\rm #1}}

\newcommand{\powerset}{\mathcal{P}}

\renewcommand{\phi}{\varphi}
\renewcommand{\emptyset}{\varnothing}

\usepackage{proof}
\usepackage[all]{xy}

\newcommand*{\rn}[1]  {\ensuremath{\mathsf{#1}}}

\newcommand*{\irn}[2][]  {\rn{#2_{I#1}}}
\newcommand*{\ern}[2][]  {\rn{#2_{E#1}}}

\begin{document}

	\begin{abstract}
	     Proof-theoretic semantics (P-tS) is the paradigm of semantics in which meaning in logic is based on proof (as opposed to truth). A particular instance of P-tS for intuitionistic propositional logic (IPL) is its \emph{base-extension semantics} (B-eS). This semantics is given by a relation called support, explaining the meaning of the logical constants, which is parameterized by systems of rules called \emph{bases} that provide the semantics of atomic propositions. In this paper, we interpret bases as collections of definite formulae and use the operational view of the latter as provided by uniform proof-search --- the proof-theoretic foundation of logic programming (LP) --- to establish the completeness of IPL for the B-eS. This perspective allows negation, a subtle issue in P-tS, to be understood in terms of the \emph{negation-as-failure} protocol in LP. Specifically, while the denial of a proposition is traditionally understood as the assertion of its negation, in B-eS we may understand the denial of a proposition as the failure to find a proof of it. In this way, assertion and denial are both prime concepts in P-tS.
	    \end{abstract}

\section{Introduction}

The definition of a system of logic may be given \emph{proof-theoretically} as a collection of rules of inference that, when composed, determine proofs; that is, formal constructions of arguments 
that establish that a conclusion is a consequence of some assumptions: 
\[
\frac{\mathrm{Established \; Premiss}_1 \quad \ldots \quad     \mathrm{Established \; Premiss}_k}{\mathrm{Conclusion}}{\big\Downarrow}
\]
The systematic use of symbolic and mathematical techniques to determine the forms of valid deductive argument defines \emph{deductive logic}: conclusions are inferred from assumptions. 

This is all very well as a way of defining what proofs are, but it relatively rarely reflects either how logic is used in practical reasoning problems or the method by which proofs are found. Rather, proofs are more often constructed by starting with a desired, or putative, conclusion and applying the rules of inference `backwards'. In this usage, the rules are sometimes called 
\emph{reduction operators}, read from conclusion to premisses, and denoted 
\[
\frac{\mathrm{Sufficient \; Premiss}_1 \quad \ldots \quad \mathrm{Sufficient \; 
    Premiss}_k}{\mathrm{Putative \; Conclusion}} {\big\Uparrow}  
\]
Constructions in a system of reduction operators are called \emph{reductions}. This paradigm is known as \emph{reductive logic}. The space of reductions of a putative conclusion is larger than its space of proofs, including also failed searches --- Pym and Ritter~\cite{pym2004reductive} have studied the reductive logic for intuitionistic and classical logic in which such objects are meaningful entities.

As one fixes more and more control structure relative to a set of reduction operators, which determining what reductions are made at what time, one increasingly delegates work to a machine. The extreme case is \emph{logic programming} (LP) in which such controls are fully specified. This view is, perhaps, somewhat obscured by the usual presentation of Horn-clause LP with SLD-resolution --- see, for example, Kowalski~\cite{Kowalski1986} and Lloyd~\cite{Lloyd1984} --- but it is explicit in work by Miller et al.~\cite{Miller1989,Miller1991}. What makes this work is that one restricts to the \emph{hereditary Harrop fragment} of a logic in which contexts contain only \emph{definite formulae} --- essentially, formulae in which disjunction only appears negatively. In LP, one typically thinks of the formulae in the context of a sequent as \emph{definional}, which underpins its use in symbolic artificial intelligence.

While deductive logic is suitable for considering the validity of propositions relative to sets of axioms, reductive logic is suitable for considering the meaning of propositions relative to \emph{systems of inference}. That the semantics of a statement is determined by its inferential behaviour is known as \emph{inferentialism} (see Brandom~\cite{Brandom2000}), which has a mathematical realization as \emph{proof-theoretic semantics} (P-tS).

In P-tS, the meaning of the logical connectives is usually derived from the rules of a natural deduction system for the logic --- for example, typically, one uses Gentzen's~\cite{Gentzen} \calculus{NJ} for intuitionistic logic. Meanwhile, the meanings of atomic propositions is supplied by an \emph{atomic system} --- a set of rules over atomic propositions. For example, taken from Sandqvist~\cite{Sandqvist2022wld}, the meaning of the proposition `Tammy is a vixen' can be understood as arising from the following rule:
\[
\infer{\mbox{Tammy is a vixen}}{\mbox{Tammy is a fox} & \mbox{Tammy is female}}
\]
 Sandqvist~\cite{Sandqvist2015IL} gave a P-tS for intuitionistic propositional logic (IPL) called \emph{base-extension semantics} (B-eS). It proceeds by a judgement called \emph{support}, parameterized by atomic systems, that defines the logical constants whose base case, the meaning of atoms, is given by derivability in an atomic system. 

There is an intuitive relationship between P-tS and LP: the way in bases are \emph{definitional} in P-tS is precisely how sets of definite formulae are \emph{definitional} in LP. Schroeder-Heister and Halln\"as~\cite{hallnas1990proof,hallnas1991proof} have used this relationship to address questions of \emph{harmony} and \emph{inversion} in P-tS. 

In this paper, we show that the completeness of IPL for the B-eS can be understood in terms of LP. The force of LP is the operational view of definite formulae, which generalize the rules in bases. Miller~\cite{Miller1989} gave this operational view of the hereditary Harrop fragment of IPL a proof-theoretic denotational semantics which proceeds by a least fixed point construction over the Herbrand base. A set of definite formulae parameterizes the construction. By thinking of this set as a base, we prove the completeness of IPL for the aforementioned B-eS by passing through the denotational semantics. 

This work exposes an interpretation of negation in P-tS as a manifestation of the \emph{negation-as-failure} (NAF) protocol.
The P-tS of negation is a subtle issue --- see, for example, K\"urbis~\cite{kurbis2019proof}. Meanwhile, in LP, the relationship between provability and refutation is made through NAF: a statement $\neg \phi$ is established precisely when the system fails to find a proof for $\phi$. The completeness argument for IPL in this paper shows that negation in B-eS can be understood in terms of the failure to find a proof. Hence, from the perspective of B-eS, it is not the case, as advanced by Frege~\cite{Frege1919} and endorsed by Dummett~\cite{dummett1991logical}, that denying a statement $\phi$ is equal to asserting the negation of $\phi$. Instead, denial in P-tS is conceptually prior to negation. In this way, through the lens of reductive logic, P-tS may be regarded as practising a form of \emph{bilateralism} --- the philosophical practice of giving equal consideration to dual concepts such as assertion and denial, truth and falsity, and so on. Of course, bilateralism with respect to negation in logic is a subject that received serious attention in the literature --- see, for example, Smiley~\cite{smiley1996rejection}, Rumfitt~\cite{rumfitt2000}, Francez~\cite{francez2014bilateralism}, Wansing~\cite{Wansing2016}, and K\"urbis~\cite{kurbis2019proof}.

The paper brings together the following fields: proof-theoretic semantics, reductive logic, and logic programming. Some such connexions have already been witnessed in the literature (see, for example, Schroeder-Heister and Halln\"as~\cite{hallnas1990proof,hallnas1991proof}). The value is that we can mutually use one to explicate phenomena in the other, such as understanding the meaning of negation in terms of NAF. That is not to argue in favour of NAF as an explanation of negation, but only that it manifests in the operational account of B-eS provided by the LP perspective.

The paper has three parts. In the first part (i.e., Section~\ref{sec:IL}), we give the relevant background on IPL: Section~\ref{sec:syntax} contains the syntax and terminology that we adopt for IPL; Section~\ref{sec:hHLP} defines the hereditary Harrop fragment (i.e., definite formulae) and gives their operational reading. In the second part (i.e., Section~\ref{sec:bes}), we summarize the B-eS for IPL as given by Sandqvist~\cite{Sandqvist2015IL}: in Section~\ref{sec:supp} we define the support relation giving the semantics, and in Section~\ref{sec:compviaBase} we summarize the existing proof of completeness. In the third part (i.e., Section~\ref{sec:LP}), we study B-eS from the perspective of the operational reading of definite formulae: Section~\ref{sec:atomicprograms} relates atomic systems and sets of definite formulae; Section~\ref{sec:compviaLP} proves completeness argument for IPL for the B-eS through the operational reading of definite formulae; and, Section~\ref{sec:NAF} discusses how this perspective manifests negation-as-failure as an explanation of the proof-theoretic meaning of negation. The paper concludes in Section~\ref{sec:conclusion} with a summary of our results and a discussion of future work.

\section{Intuitionistic Propositional Logic} \label{sec:IL}
\subsection{Syntax and Consequence} \label{sec:syntax}

There are various presentation of intuitionistic propositional logic (IPL) in the literature. We begin by fixing the relevant concepts and terminology used in this paper.

\begin{definition}[Formulae]
     Fix a (denumerable) set of atomic propositions \setAtoms. The set of formulae $\setFormulae$ (over \setAtoms) is constructed by the following grammar:
    \[
    \phi ::= \at{p} \in \setAtoms \mid \phi \lor \phi \mid \phi \land \phi \mid \phi \to \phi \mid \bot 
    \]
\end{definition}

\begin{definition}[Sequent]
     A sequent is a pair $\Gamma \seq\phi$ in which $\Gamma$ is a (countable) set of formulae and $\phi$ is a formula.
\end{definition}

We use $\proves$ as the consequence judgement relation defining IPL --- that is, $\Gamma \proves \phi$ denotes that the sequent $\Gamma \seq\phi$ is a consequence of IPL. We may write $\proves \phi$ to abbreviate $\emptyset \proves \phi$.

Throughout, we assume familiarity with the standard natural deduction system \calculus{NJ} for IPL as introduced by Gentzen~\cite{Gentzen} --- see, for example, van Dalen~\cite{vanDalen} and Troelstra and Schwichtenberg~\cite{troelstra2000basic}). Nonetheless we provide the relevant definitions in quick succession to keep the paper self-contained

\begin{definition}[Natural Deduction Argument]
    A natural deduction argument is a rooted tree of formulas in which some (possibly no) leaves are marked as discharged. An argument is open if it has undischarged assumptions; otherwise, it is closed.
\end{definition}

The leaves of an argument are its \emph{assumptions}, the root is its \emph{conclusion}. That $\argument{A}$ has open assumptions $\Gamma$, closed assumptions $\Delta$, and conclusion $\phi$ may be denoted as follows:
\[
\deduce{\phi}{\argument{A}} \qquad \deduce{\argument{A}}{\Gamma,[\Delta]} \qquad \deduce{\phi}{\deduce{\argument{A}}{\Gamma,[\Delta]}}
\]

\begin{definition}[Natural Deduction System \calculus{NJ}]
    The natural deduction system $\calculus{NJ}$ is composed of the rules in Figure~\ref{fig:nj}.
\end{definition}

\begin{figure}[t]
   \fbox{
   \begin{minipage}{0.95\linewidth}
   \centering
   $
       \infer[\irn{\land}]{\phi \land \psi}{\phi & \psi} \qquad \infer[\ern{\land^1}]{\phi}{\phi \land \psi} \quad \infer[\ern{\land^2}]{\psi}{\phi \land \psi}
    $
    \\[1.5ex]
    $
    \infer[\irn{\lor^1}]{\phi \lor \psi}{\phi} 
    \quad 
    \infer[\irn{\lor^2}]{\phi \lor \psi}{\psi} 
    \qquad
    \infer[\ern{\lor}]{\chi}{\phi \lor \psi & \deduce{\chi}{[\phi]} & \deduce{\chi}{[\psi]}}
     $
    \\[1.5ex]
    $
    \qquad \infer[\irn{\to}]{\phi \to \psi}{\deduce{\phi}{[\psi]}} 
    \qquad
    \infer[\ern{\to}]{\phi}{\phi & \phi \to \psi} 
    \qquad \infer[\ern{\bot}]{\phi}{\bot} 
    $
   \end{minipage}
   }
    \caption{Calculus $\calculus{NJ}$}
    \label{fig:nj}
\end{figure}

\begin{definition}[\calculus{NJ}-Derivation]
    The set of \calculus{NJ}-derivations is defined inductively as follows:
    \begin{itemize}
    \item[-]\textsc{Base Case.} If $\phi$ is a formula, then the one element tree $\phi$ is an \calculus{NJ}-derivation.
    
    \item[-] \textsc{Inductive Step.} Let $\rn{r}$ be a rule in $\calculus{NJ}$ and $\argument{D}_1,...,\argument{D}_n$ be a (possibly empty) list of \calculus{NJ}-derivations. If $\mathcal{D}$ is an argument arising from applying $\rn{r}$ to $\mathcal{D}_1,...,\mathcal{D}_n$, then $\mathcal{D}$ is an \calculus{NJ}-derivation.
    \end{itemize}
\end{definition}

 If $\mathcal{D}$ is an $\calculus{NJ}$-derivation with undischarged leaves composing the set $\Gamma$ and root $\phi$, then it is an argument for the sequent $\Gamma \seq \phi$. In this paper, we characterize IPL by $\calculus{NJ}$:
\[
\Gamma \proves \phi\qquad \text{iff} \qquad \mbox{there is an $\calculus{NJ}$-derivation for $\Gamma \seq \phi$}
\]

\subsection{The Hereditary Harrop Fragment} \label{sec:hHLP}

The hereditary Harrop fragment of IPL admits an operational reading that we use to deliver the completeness of a proof-theoretic semantics for IPL. This section closely follows work by Miller~\cite{Miller1989} (see also Harland~\cite{harland1991hereditary}).

The propositional hereditary Harrop formulae are generated by the following grammar in which $A \in \setAtoms$ is an atomic proposition, $D$ is a \emph{definite formula}, and $G$ is a \emph{goal} formula:
\[
\begin{array}{lcl}
     D & := & A \mid G \to A \mid D \land D  \\
     G & := & A \mid D \to G \mid G \land G \mid G \lor G
\end{array}
\]
A finite set of definite formulae $\base{P}$ is a \emph{program}; the set of all programs is $\set{P}$. We call a sequent $ \base{P} \seq G $, in which $\base{P}$ is a program and $G$ is a goal, a \emph{query}. 

The hereditary Harrop fragment of IPL admits an operational reading which renders it a logic programming language, here called hHLP. The operational semantics of hHLP is given by \emph{uniform} proof-search for $\base{P}\seq G$ in a sequent calculus for IPL --- see Miller et al.~\cite{Miller1991}.

For purely technical reasons, we require a decomposition function $[-]:\set{P} \to \set{P}$ that will unpack conjunctions. Let $[\base{P}]$ be the least set satisfying the following:
\begin{itemize}
    \item[-] $\base{P} \subseteq [\base{P}]$
    \item[-] If $D_1 \land D_2 \in [\base{P}]$, then
    $D_1 \in [\base{P}]$ and $D_2 \in [\base{P}]$.
\end{itemize}

\begin{definition}[Operational Semantics for hHLP]
    The operational semantics for hHLP is given by the clauses in Figure~\ref{fig:os}.
\end{definition}

\begin{figure}
    \centering
    \fbox{
        \begin{minipage}{0.95\linewidth}
              \[
              \begin{array}{lclr}
              \base{P} \proves A & \mbox{if} & A \in [\base{P}] & (\rn{IN})
              \\
              \base{P} \proves A & \mbox{if} & \mbox{$G \to A \in [\base{P}]$ and $\base{P} \proves G$} & (\rn{CLAUSE}) \\
              \base{P} \proves G & \mbox{if} & \base{P} \proves \bot & (\rn{EFQ}) \\
              \base{P} \proves G_1 \lor G_2 & \mbox{if} & \mbox{$\base{P} \proves G_1$ or $\base{P} \proves G_2$} & (\rn{OR}) \\
              \base{P} \proves G_1 \land G_2 & \mbox{if} & \mbox{$\base{P} \proves G_1$ and $ \base{P} \proves
              G_2$} & (\rn{AND}) \\
              \base{P} \proves D \to G & \mbox{if} & \mbox{$\base{P}\cup\{D\} \proves G$} &(\rn{LOAD}) \\
              \end{array}
              \]
        \end{minipage}
    }
    \caption{Operational Semantics for hHLP}
    \label{fig:os}
\end{figure}

Importantly, hHLP language is complete for the hereditary Harrop fragment of IPL; that is, $\base{P}\seq G$ has a successful execution iff it is a consequence of IPL --- see Miller~\cite{Miller1991}.

The standard frame semantics for IPL by Kripke~\cite{kripke1965semantical} forms a model-theoretic semantics for hHLP. However, the hereditary Harrop fragment is sufficiently restrictive that we may simplify the semantics in a useful way.

\begin{definition}[Interpretation]
 An interpretation is a mapping $I: \set{P} \to \powerset(\setAtoms)$ such that $\base{P} \subseteq \base{Q}$ implies $I(\base{P}) \subseteq I(\base{Q})$.
\end{definition}
 \begin{definition}[Satisfaction] \label{def:sat}
  The satisfaction judgement is given by the clauses of Figure~\ref{fig:satisfaction}.
 \end{definition}
\begin{figure}
    \centering
    \fbox{
        \begin{minipage}{0.95\linewidth}
              \[
              \begin{array}{lcl}
              I, \base{P} \sat A & \mbox{iff} & A \in I(\base{P})
              \\
              I, \base{P} \sat \bot & \mbox{iff} & \bot \in I(\base{P})
              \\
              I, \base{P} \sat G_1 \lor G_2 & \mbox{iff} & \mbox{$I, \base{P} \sat G_1$ or $I, \base{P} \sat G_2$}\\
              I, \base{P} \sat G_1 \land G_2 & \mbox{iff} & \mbox{$I, \base{P} \sat G_1$ and $I, \base{P} \sat G_2$}\\
              I, \base{P} \sat D \to G & \mbox{iff} & \mbox{$I, \base{P}\cup\{D\} \sat G$}\\
              \end{array}
              \]
        \end{minipage}
    }
    \caption{Denotational Semantics for hHLP}
    \label{fig:satisfaction}
\end{figure}

We desire a particular interpretation $J$ such that the following holds:
\[
J, \base{P} \entails G \qquad \text{ iff } \qquad \base{P} \proves G
\]

To this end, we consider a function $T$ from interpretations to interpretations that corresponds to unfolding derivability in a base:
\[
\begin{array}{lcl}
    T(I)(\base{P}) & := &  \{ A  \mid A \in [\base{P}] \} \, \cup \\
    & & \{ A \mid \mbox{$(G \to A) \in [\base{P}]$ and $I, \base{P} \entails G$} \} \, \cup \\
     & & \{A \mid I, \base{P} \entails \bot \}
\end{array}
\]

Interpretations form a lattice under point-wise union ($\sqcup$), point-wise intersection ($\sqcap$), and point-wise subset ($\sqsubseteq$); the bottom of the lattice is given by $I_\bot: \base{P} \mapsto \emptyset$. It is easy to see that $T$ is monotonic and continuous on this lattice, and, by the Knaster-Tarski Theorem~\cite{apt1982contributions}, its least fixed-point is given as follows:
\[
T^\omega I_\bot := I_\bot \sqcup T(I_\bot) \sqcup T^2 (I_\bot) \sqcup \hdots
\]
Intuitively, each application of $T$ concerns the application of a clause so that $T^\omega I_\bot$ corresponds to arbitrarily many applications. 

\begin{lemma} \label{lem:miller}
For any program $\base{P}$ and goal $G$,
\[
T^\omega I_\bot, \base{P} \entails G \qquad \text{ iff } \qquad  \base{P} \proves G
\]
\end{lemma}
\begin{proof}
 The result was proved by Miller~\cite{Miller1989} --- see also Harland~\cite{harland1991hereditary}. 
\end{proof}

\section{Base-extension Semantics} \label{sec:bes}

In this section, we give a brief, but complete, synopsis of the base-extension semantics (B-eS) for IPL as introduced by Sandqvist~\cite{Sandqvist2015IL}. The semantics proceeds through a \emph{support} relation parameterised by certain atomic systems, called \emph{bases}. We differ slightly in presentation from the previous work: first, we refer to more the possibility of more general definitions (e.g., considering $n$th level atomic systems for $n>2$); second, we make use of derivations as mathematical objects; third, we parameterize support over a notion of base called a \emph{basis}, a classes of atomic systems. These difference help bridge the gap between the earlier work and the connexions to logic programming in this paper. It also sets the B-eS for IPL within the wider literature of P-tS from which we draw the generalizations. 

\subsection{Support in a Base} \label{sec:supp}

A common idea in proof-theoretic semantics --- the paradigm of meaning in which B-eS operates --- is that the meaning of atomic propositions is given by sets of atomic rules governing their inferential behaviour.  Piecha and Schroeder-Heister~\cite{Schroeder2016atomic,piecha2017definitional} have given a useful inductive hierarchy of them.

\begin{definition}[Atomic Rule] \label{def:atomicrule}
 An $n$th-level atomic rule is defined as follows:
 \begin{itemize}
  \item[-] A zeroth-level atomic rule is a rule of the following form in which $\at{c} \in \setAtoms$:
 \[
 \infer{\,\at{c}\,}{}
 \]
 \item[-] A first-level atomic rule is a rule of the following form in which $\at{p}_1,...,\at{p}_n,\at{c} \in \setAtoms$,
 \[
 \infer{\,\at{c}\,}{\, \at{p}_1 & \hdots & \at{p}_n \,}
 \]
 \item[-] An $(n+1)$th-level atomic rule is a rule of the following form in which $\at{p}_1,...,\at{p}_n,\at{c} \in \setAtoms$ and $\Sigma_1,...,\Sigma_n$ are (possibly empty) sets of $n$th-level atomic rules:
  \[
 \infer{\,\at{c}\,}{ \, \deduce{\at{p}_1}{[\Sigma_1]} & \hdots & \deduce{\at{p}_n}{[\Sigma_n]} \,}
 \]
 \end{itemize}
\end{definition}

We take that premisses may be empty such that an $m$th-level atomic rule is an $n$th-level atomic rule for any $n>m$. Having sets of atomic rule as hypotheses is more general than have sets of atomic propositions as hypotheses; the latter is captured by the former by taking zeroth-order atomic rules. Nonetheless, the generalization is, perhaps, unexpected. We discuss it further in Section~\ref{sec:compviaLP}.

\begin{definition}[Atomic System]
    An atomic system is a set of atomic rules.
\end{definition}

Atomic systems may have infinitely many rules but they are at most countably infinite. They are used to base validity in P-tS on proof. The definition of a derivation is a generalization of natural deduction \emph{\`a la} Gentzen~\cite{Gentzen}, which was given by Piecha and Schroeder-Heister~\cite{Schroeder2016atomic,piecha2017definitional}. 

\begin{definition}[Derivation in an Atomic System]
 Let $\base{A}$ be an atomic system. The set of $\base{A}$-derivations is defined inductive as follows:
 \begin{itemize}
     \item[-] \textsc{Base Case}. If $\base{A}$ contains a zeroth-level rule concluding $\at{c}$, then the natural deduction argument consisting of just the node $\at{c}$ is a $\base{A}$-derivation.
     \item[-] \textsc{Induction Step}. Suppose $\base{A}$ contains an $(n+1)$th-level rule $\rn{r}$ of the following form:
     \[
 \infer{\,\at{c}\,}{ \, \deduce{\at{p}_1}{[\Sigma_1]} & \hdots & \deduce{\at{p}_n}{[\Sigma_n]} \,}
 \]
 And suppose that for each $1 \leq i \leq n$ there is a $\base{A}$-derivation $\argument{D}_i$ of the following form:
 \[
 \deduce{\at{p}_i}{\deduce{\mathcal{D}_i}{\Gamma_i, \Sigma_i}}
 \]
 Then the natural deduction argument with root $\at{c}$ and immediate sub-trees $\argument{D}_1$,...,$\argument{D}_n$ is a $\base{A}$-argument of $\at{c}$ from $\Gamma_1\cup...\cup \Gamma_n \cup \base{A}$.
 \end{itemize}
 An atom $\at{c}$ is derivable from $\Gamma$ in $\base{A}$ --- denoted $\Gamma \proves_\base{A} \at{c}$ --- iff there is a $\base{A}$-derivation of $\at{c}$ from $\Sigma \cup \base{A}$.
 \end{definition}

Typically, we do not consider all atomic systems, but restrict attention to some particular class.

\begin{definition}[Basis]
    A basis is a set of atomic systems.
\end{definition}

Having fixed a basis $\basis{B}$, an atomic system $\base{B} \in \basis{B}$ is called a \emph{base}. A base-extension semantics is formulated relative to a basis via a support relation.

\begin{definition}[Support in a Base] \label{def:bvalid}
 Fix a basis $\basis{B}$. Support over $\base{B}$ is the least relation $\supp_{-}$ on sequents  and bases in $\basis{B}$ defined by the clause of Figure~\ref{fig:support}.  The  validity judgement over $\basis{B}$ is the following relation $\supp$ one sequent: 
 \[
 \Gamma \supp \phi \qquad \mbox{iff} \qquad \mbox{$\Gamma \supp_\base{B} \phi$ for any $\base{B} \in \basis{B}$}
 \]
\end{definition}
\begin{figure}
    \centering
    \fbox{
    \begin{minipage}{0.95\linewidth}
          \[
        \begin{array}{lclr}
         \Gamma \supp_{\base{B}} \phi & \text{ iff } & \text{for any $\base{C} \in \basis{B}$ such that $\base{B} \subseteq \base{C}$,} & \text{($\Rightarrow$)} \\ 
         & & \text{if $  \supp_{\base{C}} \psi$ for all $\psi \in \Gamma$, then $  \supp_{\base{C}} \phi$ } &  \\ 
            \supp_{\base{B}} \at{p}  & \text{ iff } &   \proves_\base{B} \at{p}  & \text{($\setAtoms$)} \\
            \supp_{\base{B}} \phi \to \psi & \text{ iff } & \phi \supp_{\base{B}} \psi &  \text{($\to$)}\\
            \supp_{\base{B}} \phi \land \psi   & \text{ iff } &   \supp_{\base{B}} \phi \text{ and }   \supp_{\base{B}} \psi &  \text{($\land$)}  \\
            \supp_{\base{B}} \phi \lor \psi & \text{ iff } &  \text{for any $\base{C} \in \basis{B}$ such that $\base{B} \subseteq \base{C}$ and} &  \text{($\lor$)} \\ & & \mbox{any $\at{p} \in \setAtoms$, if $\phi \supp_\base{C} \at{p} \text{ and } \psi \supp_\base{C} \at{p}$, then $\supp_\base{C} \at{p}$} & \\
          \supp_{\base{B}} \bot & \text{iff} &    \supp_\base{B} \at{p} \text{ for any $\at{p} \in \setAtoms$} &  \text{($\bot$)} \\
        \end{array}
        \]
    \end{minipage}
    }
    \caption{Support in a Base}
    \label{fig:support}
\end{figure}

Observe that $\supp_\base{B} \phi$ coincides with $\emptyset \supp_\base{B} \phi$. Symmetrically, we write $ \supp \phi$ to denote $\emptyset \supp \phi$. 

Sandqvist~\cite{Sandqvist2005inferentialist} gave this semantics with a basis $\basis{S}$ consisting of atomic rules that are \emph{properly} second-level; that is, rules of the form
\[
\infer{\at{c}}{\deduce{\at{p}_1}{[\Sigma_1]} & \hdots & \deduce{\at{p}_n}{[\Sigma_n]} }
\]
in which $\Sigma_1$,...,$\Sigma_n$ are sets of atoms. 

\begin{theorem}[Soundness \& Completeness] \label{thm:sandqvist}
$\Gamma \proves \phi$ iff $\Gamma \supp \phi$ over $\basis{S}$.
\end{theorem}
\begin{proof}
Proved by Sandqvist~\cite{Sandqvist2015IL} --- see Section~\ref{sec:compviaBase}.
\end{proof}

The support relation satisfies some important expected properties, such as the following:

\begin{lemma}\label{lem:extensionisconservative}
 If $\Gamma \supp_{\base{B}} \phi$ and $\base{C} \supseteq \base{B}$, then  $\Gamma \supp_{\base{C}} \phi$.
\end{lemma}
\begin{proof}
Proved by Sandqvist~\cite{Sandqvist2015IL} by induction on support in a base.
\end{proof}

There are related base-extension semantics for classical logic --- see Sandqvist~\cite{Sandqvist2005inferentialist,Sandqvist2009CL} and Makinson~\cite{makinson2014inferential}. 

This summarizes the B-eS for IPL. In the next section we present the completeness proof as provided by Sandqvist~\cite{Sandqvist2015IL} as it will be useful to understand the connections to reductive logic later on.

\subsection{Completeness of IPL via a Natural Base} \label{sec:compviaBase}

Sandqvist~\cite{Sandqvist2015IL} proved the soundness of IPL for the B-eS by showing that validity admits all the rules of \calculus{NJ}. His proof of completeness is more complex. In essence, Sandqvist~\cite{Sandqvist2015IL} proved completeness of IPL for the B-eS by constructing a bespoke atomic system $\base{N}$ to a given validity judgement that allows us to \emph{simulate} an \calculus{NJ}-derivation for the sequent in question. We present the main ideas here as we refer to them in Section~\ref{sec:compviaLP}.

We want to show that if $\Gamma \supp \gamma$ obtains, then there is an $\calculus{NJ}$-proof witnessing $\Gamma \proves \gamma$. To this end, we associate to each formula $\phi$ in the sequent $\Gamma \seq\gamma$ a unique atom $r$ and construct a base $\base{N}$ emulating \calculus{NJ} such that $r$ behaves in $\base{N}$ as $\phi$ behaves in $\calculus{NJ}$. For example, let $\Gamma \seq\gamma$ contain $\phi:=\at{p} \land \at{q}$. The rules governing $\phi$ are the conjunction introduction and elimination rules of $\calculus{NJ}$, so we require \base{N} to contain the following rules in which $\at{r}$ is alien to $\Gamma \seq\gamma$:
\[
\infer{\,\at{r}\,}{\,\at{p} & \at{q} \,} \qquad \infer{\,\at{p}\,}{\,\at{r}\,} \quad \infer{\,\at{q}\,}{\,\at{r}\,}
\]
These rules are designed such that $\at{r}$ behaves in $\base{N}$ precisely as $\phi$ does in $\calculus{NJ}$; that is, they emulate the conjunction rules. The shorthand for $\at{r}$ is $(\at{p} \land \at{q})^\flat$ --- that is $r=\phi^\flat$ --- so that the above rules may be expressed more clearly as follows:
\[
\infer{\,(\at{p}\land\at{q})^\flat\,}{\,\at{p} & \at{q}\,} \qquad \infer{\,\at{p}\,}{\,(\at{p}\land\at{q})^\flat\,} \quad \infer{\,\at{q}\,}{\,(\at{p}\land\at{q})^\flat\,}
\]
For clarity, we give another example. Suppose $\Gamma \seq\gamma$ also contains $\psi := \at{p} \to \at{q}$, then $\base{N}$ contains rules that emulate the implication introduction and elimination rules of $\calculus{NJ}$ for $\psi$ using an atom $\psi^\flat = (\at{p} \to \at{q})^\flat$ alien to $\Gamma$ and $\gamma$. That is, $\base{N}$ contains the following rules:
\[
\infer{(\at{p} \to \at{q})^\flat}{\deduce{\at{q}}{[\at{p}]}} \qquad 
\infer{\at{q}}{\at{p} & (\at{p}\to\at{q})^\flat} 
\]
The details of how $\base{N}$ is constructed and how it delivers completeness are below. 

Given $\Gamma \supp \gamma$, to every formula $\phi$ occurring in $\Gamma \seq\gamma$ associate a unique atomic proposition $\phi^\flat$ as follows:
\begin{itemize}
    \item[-] if $\phi \not \in \setAtoms$, then $\phi^\flat$ is an atom that does not occur in $\Gamma \seq\gamma$;
    \item[-] if $\phi \in \setAtoms$, then $\phi^\flat = \phi$.
\end{itemize} 
The right-inverse of $-^\flat$ is $-^\natural$ and both functions act on sets point-wise,
\[
\Sigma^\flat := \{\phi^\flat \mid \phi \in \Sigma \} \qquad \Sigma^\natural := \{\phi^\natural \mid \phi \in \Sigma \}
\]

Let $\base{N}$ be the atomic system containing precisely the rules of Figure~\ref{fig:baseN} for any $\phi$, $\psi$, and $\chi$ occurring in $\Gamma \seq\gamma$. These rules are precisely such that $\phi^\flat$ behaves in $\base{N}$ as $\phi$ does in \calculus{NJ}. Note that, for any validity judgement, the atomic system $\base{N}$ thus generated is indeed a Sandqvist base; moreover, it is a finite set.

\begin{figure}[t]
   \fbox{
   \begin{minipage}{0.95\linewidth}
   \centering
   \vspace{0.5em}
   $
       \infer[\irn{\land}^\flat]{(\phi \land \psi)^\flat}{\phi^\flat & \psi^\flat} \qquad \infer[\ern{\land}^\flat]{\phi^\flat}{(\phi \land \psi)^\flat} \quad \infer[\ern{\land}^\flat]{\psi^\flat}{(\phi \land \psi)^\flat}
    $
    \\[1.5ex]
    $
    \infer[\irn{\lor}^\flat]{(\phi \lor \psi)^\flat}{\phi^\flat} 
    \quad 
    \infer[\irn{\lor}^\flat]{(\phi \lor \psi)^\flat}{\psi^\flat} 
    \qquad
    \infer[\ern{\lor}^\flat]{\chi^\flat}{(\phi \lor \psi)^\flat & \deduce{\chi^\flat}{[\phi^\flat]} & \deduce{\chi^\flat}{[\psi^\flat]}}
    $
     \\[1.5ex]
    $
     \infer[\irn{\to}^\flat]{(\phi \to \psi)^\flat}{\deduce{\psi^\flat}{[\phi^\flat]}}
     \qquad 
         \infer[\ern{\to}^\flat]{\psi^\flat}{\phi^\flat & (\phi \to \psi)^\flat} 
         \qquad 
           \infer[\ern{\bot}^\flat]{\phi^\flat}{\bot^\flat}
    $
    \vspace{0.5em}
   \end{minipage}
   }
    \caption{Atomic System $\base{N}$}
    \label{fig:baseN}
\end{figure}

In this set-up, Sandqvist~\cite{Sandqvist2015IL} establishes three properties that collectively deliver completeness.

\begin{lemma}\label{lem:basicocmpleteness}
    Let $\Sigma \subseteq \mathbb{A}$ and $\at{p} \in \mathbb{A}$ and let $\base{B} \in \basis{S}$,
    \[
    \Sigma \supp_\base{B} \at{p} \qquad \text{ iff } \qquad \Sigma \proves_\base{B} \at{p}
    \]
\end{lemma}

This claim is a basic completeness result in which the context $\Sigma$ is restricted to a set of atomic propositions and the extract $\at{p}$ is an atomic proposition. 

\begin{lemma}\label{lem:flatequivalence}
    For every $\phi$ occurring in $\Gamma \seq\gamma$ and any $\base{N}' \supseteq \base{N}$,
    \[
     \supp_{\base{N}'} \phi^\flat \qquad  \text{ iff } \qquad  \supp_{\base{N}'} \phi
    \]
\end{lemma}

In other words, $\phi^\flat$ and $\phi$ are equivalent in $\base{N}$ --- that is, $\phi^\flat \supp_\base{N} \phi$ and $\phi \supp_\base{N} \phi^\flat$. The property allows us to move between the basic case (i.e., the set-up of Lemma~\ref{lem:basicocmpleteness}) and the general case (i.e., completeness --- Theorem~\ref{thm:sandqvist}). This is the crucial step in the proof of completeness. In Section~\ref{sec:compviaLP}, we study it in terms of the operational account of definite formulae given in Section~\ref{sec:hHLP}.

\begin{lemma} \label{lem:sharpening}
    Let $\Sigma \subseteq \mathbb{A}$ and $\at{p} \in \mathbb{A}$,
    \[
    \Sigma \supp_\base{N} p \text{ implies } \Sigma^\natural \proves p^\natural
    \]
\end{lemma}

This property is the simulation statement. It allows us to make the final move from derivability in $\base{N}$ to derivability in $\calculus{NJ}$.

These lemmas collectively suffice for completeness:

\begin{proof} Theorem~\ref{thm:sandqvist} --- Completeness. If $\Gamma\supp \chi$, then $\Gamma^\flat \supp_\base{N} \chi^\flat$ because if $\base{N}' \supseteq \base{N}$ and $\emptyset \supp_{\base{N}'} \phi^\flat$ for $\phi^\flat \in \Gamma^\flat$, then (by Lemma~\ref{lem:flatequivalence}) $\emptyset \supp_{\base{N}'} \phi$ for every $\phi \in \Gamma$. Hence, $\emptyset \supp_{\base{N}'} \chi$ (since $\Gamma \supp \chi$); whence (by Lemma~\ref{lem:flatequivalence}) $\emptyset \supp_{\base{N}'} \chi^\flat$; whence (by Lemma~\ref{lem:basicocmpleteness}) it follows that $\Gamma^\flat \supp_{\base{N}} \chi^\flat$. Thus (by Lemma~\ref{lem:sharpening}) it follows that $\Gamma \proves \chi$.
\end{proof} 

In the next section, we show that the completeness follows intuitively from regarding $\base{N}$ as a program capturing the inferential content of \calculus{NJ}. In general, a base may be regarded as a program, so that the application of a rule in the base corresponds to the use of a clause in the program. We demonstrate that the validity of a formula $\phi$ in the base $\base{N}$ emulates the execution of a goal $\phi^\flat$ relative to the program $\base{N}$. By construction of $\base{N}$, such executions simulate the construction of an $\calculus{NJ}$ proof of $\phi$. Hence, IPL is complete with respect to the B-eS. 

\section{Definite Formulae, Proof-search, and Completeness} \label{sec:LP}

There is an intuitive encoding of atomic rules as formulae. More precisely, as \emph{definite} formulae. Under this encoding, the bases which deliver B-eS live within the hereditary Harrop fragment of IPL. The latter has a simple operational reading via proof-search for uniform proofs (see Section~\ref{sec:hHLP}) that enables a proof-theoretic denotational semantics --- the least fixed point construction. We use this well-understood phenomenon to deliver the completeness of IPL with respect to Sandqvist's B-eS~\cite{Sandqvist2015IL} --- see Section~\ref{sec:bes}. 

Doing this reveals a subtle interpretation of the meaning of negation in terms of the negation-as-failure protocol. A reductive logic view of the denial of a formula is the failure to find a proof of it. Thus, according to the view of B-eS arising from the account passing through the operational reading of definite formulae, in B-eS denial is conceptionally prior to negation and both require equal consideration.
 
\subsection{Atomic Systems vs. Programs} \label{sec:atomicprograms}

Intuitively, atomic systems in B-eS are definitional in precisely the same way as programs in hHLP are definitional. To illustrate this, we must systematically move between them, which we do by encoding atomic systems as programs. 

Let $\lfloor - \rfloor$ be as follows:
\begin{itemize}
    \item[-] The encoding of zeroth-level rule is as follows:
    \[
   \left\lfloor 
    \raisebox{-1ex}{
    \infer{\,\at{c}\,}{\,}
       }
        \right\rfloor
        := \at{c}
    \]
    \item[-] The encoding of a first-level rule is as follows:
    \[
    \left\lfloor 
    \raisebox{-1ex}{
    \infer{\,\at{c}\,}{\at{p}_1  & \hdots & \at{p}_n} 
    }
        \right\rfloor
    :=
    (\at{p}_1 \land \hdots \land \at{p}_n) \to \at{c}
    \]
    \item[-] The encoding of an $n$th-level rule is as follows:
    \[
\left\lfloor
\raisebox{-1em}{
\infer{\,\at{c}\,}{\deduce{\at{p_1}}{[\Sigma_1]} \, \hdots \, \deduce{\at{p}_n}{[\Sigma_n]}} 
}
\right\rfloor :=
\big((\lfloor \Sigma_1 \rfloor \to \at{p}_1) \, \land \, \hdots \, \land \, (\lfloor \Sigma_n \rfloor \to \at{p}_n)\big) \to \at{c}
\]
\end{itemize}

The hierarchy of atomic system provided by Piecha and Schroeder-Heister~\cite{Schroeder2016atomic,piecha2017definitional} (Definition~\ref{def:atomicrule}) precisely corresponds to the inductive depth of the grammar for hereditary Harrop formulae --- that is, if $\base{A}$ is an $n$-th level atomic system, then
\[
 \proves_\base{A} \at{p} \qquad \mbox{ iff } \qquad \lfloor \base{A} \rfloor \proves \at{p}
\]
Therefore, we may suppress the encoding function, and henceforth use atomic systems and programs interchangeably.

Of course, in the Sanqvist basis, we are limited to properly second-level atomic systems, but the grammar of definite clauses can handle considerably more. Indeed, the work below suggests that completeness holds for $n$th-level atomic systems for $n\geq2$. 

Formally, to say that bases are definitional in the sense of programs, we mean the following:
\[
\supp_\base{B} \phi \qquad \text{ iff } \qquad \base{N}\cup \base{B} \proves \phi^\flat \tag{$\ast$}
\]
We assume for this equivalence that $-^\flat$ is sensitive to the presence of $\base{B}$ so that $\phi^\flat$ does not occur in $\base{B}$ for $\phi \neq \setAtoms$. That we use $\phi^\flat$ rather than $\phi$ in the $(\ast)$ is essentially. It is certainly \emph{not} the case that bases behave exactly as contexts; that is, we do \emph{not} have the following equivalence:
\[
\supp_\base{B} \phi \qquad \text{ iff } \qquad \base{B} \proves \phi \tag{$\ast\ast$}
\]

That this generalisation fails is shown by the following counter-example.

\begin{example} \label{ex:counter}
Consider the following formula:
\[
\phi:= (\at{a} \to \at{b} \lor \at{c}) \to \big((\at{a}\to \at{b}) \lor (\at{a} \to \at{c}) \big)
\]
The formula $\phi$ is not a consequence of IPL; hence, by completeness of IPL with respect to the B-eS, $\supp_\base{B} (\at{a} \to \at{b} \lor \at{c})$ and $\not \supp_\base{B}(\at{a}\to \at{b}) \lor (\at{a} \to \at{c})$, for some $\base{B}$. However, assuming $(\ast\ast)$, we have the following:
\[
\begin{array}{lclr}
\supp_\base{B} \at{a} \to \at{b} \lor \at{c} & \text{ implies } & \base{B} \proves \at{a} \to \at{b} \lor \at{c} & (\ast\ast) \\
  & \text{ implies } & \base{B} \cup \{\at{a}\} \proves \at{b} \lor \at{c} & (\rn{LOAD}) \\
    & \text{ implies } & \base{B} \cup \{\at{a}\} \proves \at{b} \text{ or }  \base{B} \cup \{\at{a}\} \proves \at{c} & (\rn{OR}) \\
      & \text{ implies } & \base{B}  \proves \at{a} \to \at{b} \text{ or }  \base{B} \proves \at{a} \to \at{c} & (\rn{LOAD}) \\
      & \text{ implies } & \base{B}  \proves (\at{a} \to \at{b})\lor(\at{a} \to \at{c}) &(\rn{OR}) \\
      & \text{ implies } & \supp_\base{B} (\at{a} \to \at{b})\lor( \at{a} \to \at{c}) & (\ast\ast)
\end{array}
\]
That is,
$\supp_\base{B} (\at{a} \to \at{b} \lor \at{c})$ implies $\supp_\base{B}(\at{a}\to \at{b}) \lor (\at{a} \to \at{c})$, for any $\base{B}$. This is a contradiction, therefore $(\ast\ast)$ fails.
\end{example}

In the next section, we use the relationship between atomic systems and programs to prove completeness of IPL with respect to the B-eS.

\subsection{Completeness of IPL via Logic Programming} \label{sec:compviaLP}
We may prove completeness of IPL with respect to the B-eS by passing through hHLP as follows:
\[
\xymatrix{
T^\omega I_\bot, \base{N} \entails \phi^\flat \ar@{<->}[r] &  \base{N} \proves \phi^\flat \ar[d] \\
 \supp_\base{N} \phi \ar[u] &  \proves \phi 
}
\]
The diagram requires three claims, the middle one of which is Lemma~\ref{lem:miller}. The other two are Lemma~\ref{lem:emulation} and Lemma~\ref{lem:simulation}, respectively, reading in the direction of the arrows.

The intuition of the completeness argument is two-fold: firstly, that $\base{N}$ is to $\phi^\flat$ as $\calculus{NJ}$ is to $\phi$; secondly, the use of a rule in a base corresponds to the use of a clause in the corresponding program; thirdly, execution in $\base{N}$ corresponds to proof(-search) in $\calculus{NJ}$.  In this set-up, the $T^\omega$ construction captures the construction of a proof: the application of a rule corresponds to a use of $T$, the iterative application of rules corresponds to the iterative application of $T$ --- that is, to $T^\omega$.

It remains to prove the claims and completeness. 

\begin{lemma}[Emulation] \label{lem:emulation}
 If $\supp_\base{N} \phi$, then $T^\omega I_\bot, \base{N} \entails \phi^\flat$.
\end{lemma}
\begin{proof}
We prove a stronger proposition: for any $\base{N}' \supseteq \base{N}$, if $\supp_{\base{N}'} \phi$, then $T^\omega I_\bot, \base{N}' \entails \phi^\flat$. We proceed by induction on support in a base according to the various cases of Figure~\ref{fig:support}, although for the sake of economy we combine the clauses $\Rightarrow$ and $\to$.
\begin{itemize}
    \item[-] $\phi \in \setAtoms$. Note $\phi^\flat = \phi$, by definition. Therefore, if $\supp_{\base{N}'} \phi$, then $\proves_{\base{N}'} \phi$, but this is precisely emulated by application of $T$. Hence, $T^\omega I_\bot, {\base{N}'}  \sat \phi$.
    
    \item[-] $\phi = \bot$. If $\supp_{\base{N}'} \bot$, then  $\supp_{\base{N}'} \at{p}$, for every $\at{p} \in \setAtoms$. By the induction hypothesis (IH), $T^\omega I_\bot, {\base{N}'} \sat \at{p}$ for every $\at{p} \in \setAtoms$. It follows that $T^\omega I_\bot, {\base{N}'} \sat \bot^\flat$.
    
    \item[-] $\phi := \phi_1 \land \phi_2$. By the $\land$-clause for support, $\supp_{\base{N}'} \phi_1$ and $\supp_{\base{N}'} \phi_2$. Hence, by the IH, $T^\omega I_\bot, {\base{N}'} \sat \phi_1$ and $T^\omega I_\bot, {\base{N}'} \sat \phi_2$. The result follows by $\land$-clause for satisfaction. 
    
      \item[-] $\phi := \phi_1 \lor \phi_2$.  By the IH, $\phi_1 \supp_{\base{N'}} \phi_1^\flat$  and $\phi_2 \supp_{\base{N'}} \phi_2^\flat$. By the $\irn{\lor}$-scheme in $\base{N}'$, both $\phi_1^\flat \supp (\phi_1\lor \phi_2)^\flat $ and $\phi_2^\flat \supp (\phi_1\lor \phi_2)^\flat $. By $\Rightarrow$-clause for support, we have $\phi_1 \supp_{\base{N'}} (\phi_1 \lor \phi_2)^\flat$  and $\phi_2 \supp_{\base{N'}} (\phi_1 \lor \phi_2)^\flat$. Since $\supp_{\base{N}'} \phi_1 \lor \phi_2$, it follows from $\lor$-clause for support that $\supp_{\base{N}'} (\phi_1 \lor \phi_2)^\flat$. That is, $(\phi_1 \lor \phi_2)^\flat \in T^\omega, \base{N}' \sat (\phi_1 \lor \phi_2)^\flat$, as required.
      
      \item[-] $\phi:=\phi_1 \to \phi_2$. We first prove the following auxiliary proposition: for any $\phi$, the judgement $\supp_{{\base{N}'} \cup \{ \phi^\flat \}} \phi$ obtains. We proceed by sub-induction on support in a base according to the various cases of Figure~\ref{fig:support}. As above, for the sake of economy we combine the clauses $\Rightarrow$ and $\to$.
    \begin{itemize}
    \item[-]  $\phi \in \setAtoms$.  The result is immediate since $\supp_{{\base{N}'} \cup \{ \phi^\flat \}} \phi$ iff ${\base{N}'} \cup \{ \phi^\flat \} \proves \phi$ and the latter obtains by \rn{IN}.
        \item[-] $\phi = \bot$. By $\bot^\flat$-scheme, $\proves_{{\base{N}'} \cup \{ \phi^\flat \}} \at{p}$ for any $\at{p} \in \set{A}$. That is, $\supp_{{\base{N}'} \cup \{ \phi^\flat \}} \at{p}$ for any $\at{p} \in \set{A}$. Thus, $\supp_{{\base{N}'} \cup \{ \phi^\flat \}} \bot$, as required.
            \item[-] $\phi = \phi_1 \land \phi_2$. By the sub-induction hypothesis (sub-IH), $\supp_{{\base{N}'} \cup \{\phi_1^\flat\}} \phi_1$ and $\supp_{{\base{N}'} \cup \{\phi_2^\flat\}} \phi_2$ obtain. By Lemma~\ref{lem:extensionisconservative}, therefore  $\supp_{{\base{N}'}\cup\{\phi_1^\flat,\phi_2^\flat\}} \phi_1$ and $\supp_{{\base{N}'}\cup\{\phi_1^\flat,\phi_2^\flat\}} \phi_2$ obtain. By Definition~\ref{def:bvalid}, we have $\supp_{{\base{N}'}\cup\{\phi_1^\flat,\phi_2^\flat\}} \phi_1 \land \phi_2$. By $\irn{\land}^\flat$- and $\ern{\land}^\flat$-schemes, $\supp_{{\base{N}'}\cup\{(\phi_1\land\phi_2)^\flat\}} \phi_1 \land \phi_2$.
    \item[-] $\phi = \phi_1 \lor \phi_2$. By the sub-IH, both $\supp_{{\base{N}'}\cup\{\phi_1^\flat\}} \phi_1$ and $\supp_{{\base{N}'}\cup\{\phi_2^\flat\}} \phi_2$ obtain.  By $\irn{\lor}^\flat$- and $\ern{\lor}^\flat$-schemes, $\supp_{{\base{N}'}\cup\{(\phi_1\lor\phi_2)^\flat\}} \phi_1$ and $\supp_{{\base{N}'}\cup\{(\phi_1\lor\phi_2)^\flat\}} \phi_2$. Therefore, $\supp_{{\base{N}'}\cup\{\phi_1^\flat,\phi_2^\flat\}} \phi_1 \lor \phi_2$ obtains.
    \item[-] $\phi = \phi_1 \to\phi_2$. 
    By the IH, if $\base{C}$ is such that $\supp_{{\base{N}'}\cup\{\phi_1^\flat \to \phi_2^\flat\}\cup \base{C}} \phi_1$, then ${\base{N}'}\cup\{(\phi_1^\flat \to \phi_2^\flat)\}\cup \base{C} \proves \phi_1^\flat$.  By the sub-IH, both $\supp_{{\base{N}'}\cup\{\phi_1^\flat\}} \phi_1$ and $\supp_{{\base{N}'}\cup\{\phi_2^\flat\}} \phi_2$ obtain. Hence, for any $\base{C} \supseteq \base{B}$, if $\supp_{{\base{N}'}\cup\{\phi_1^\flat\to \phi_2^\flat\}\cup \base{C}} \phi_1$, then $\supp_{{\base{N}'}\cup\{\phi_1^\flat\to \phi_2^\flat\}\cup \base{C}} \phi_2$. By Definition~\ref{def:bvalid}, we have $\supp_{{\base{N}'}\cup\{\phi_1^\flat\to \phi_2^\flat\}} \phi_1 \to \phi_2$. By $\irn{\to}^\flat$- and $\ern{\to}^\flat$-schemes, $\supp_{{\base{N}'}\cup\{(\phi_1\to \phi_2)^\flat\}} \phi_1 \to \phi_2$.
\end{itemize}
This completes the sub-induction. It remains to consider the case for $\to$-clause for the main induction. 

      By the $\to$-clause for satisfaction, $\phi_1 \supp_\base{N} \phi_2$. So, by the $\Rightarrow$-clause for satisfaction, $\supp_{\base{N}'}\phi_1$ implies $\supp_{\base{N}'} \phi_2$  for any $\base{N}' \supseteq \base{N}$. In particular, let $\base{N}':= \base{N} \cup \{\phi_1^\flat\}$. Since $\supp_{\base{N}'}\phi_1$ obtains by the sub-induction, we have $\supp_{\base{N}'} \phi_2$. By the IH, $T^\omega I_\bot, \base{N}\cup\{\phi_1^\flat\} \sat \phi_2^\flat$. Hence, $T^\omega I_\bot, \base{N} \sat \phi_1^\flat \to \phi_2^\flat$. By construction of $\base{N}$, we have $(\phi_1^\flat \to \phi_2^\flat)\to(\phi_1\to\phi_2)^\flat \in \base{N}$. Therefore, by definition of $T$, we have $(\phi_1 \to \phi_2)^\flat \in T(T^\omega I_\bot)(\base{N})$. Whence, $T^\omega I_\bot \base{N} \proves (\phi_1 \to \phi_2)^\flat$, as required.
\end{itemize}

This completes the induction.
\end{proof}

\begin{lemma}[Simulation] \label{lem:simulation}
 If $\base{N} \proves \phi^\flat$, then $\proves \phi$.
\end{lemma}
\begin{proof}
We proceed by induction on the length of execution. A more tractable induction invariant is the following: if $\base{N}\cup\Gamma^\flat \proves \phi^\flat$, then $\Gamma \proves \phi$. Intuitively, the execution of $\base{N} \cup \Gamma^\flat  \proves \phi^\flat$ simulates the reductive construction of a proof of $\phi$ from $\Gamma$ in $\calculus{NJ}$ --- that is, a proof-search. We proceed by induction on the length of the execution.

 \textsc{Base Case:} It must be that $\phi \in \Gamma$, so $\Gamma \proves \phi$ is immediate.
 
 \textsc{Inductive Step:} By construction of $\base{N}$, the execution concludes by \textsf{CLAUSE} applied to a definite clause $\rho$ simulating a rule $\at{r} \in \calculus{NJ}$; that is, $\base{N}\cup \Gamma^\flat \proves \psi_i^\flat$ for $\psi_i$ such that $\psi_1^\flat \land .... \land \psi_n^\flat \to \phi^\flat$. By the induction hypothesis (IH), $\Gamma \proves \psi_i$ for $1 \leq i \leq n$. It follows that $\Gamma \proves \phi$ by applying $r \in \calculus{NJ}$.
 
 For example, if the execution concludes by \textsf{CLAUSE} applied to the clause for $\land$-introduction (i.e., $\phi^\flat \land \psi^\flat \to (\phi\land\psi)^\flat$), then the trace is as follows:
 \[
\infer{\base{N} \proves (\phi \land \psi)^\flat}{\infer{\base{N} \proves \phi^\flat \land \psi^\flat}{\deduce{\base{N} \proves \phi^\flat}{\vdots} & \deduce{\base{N}\proves \psi^\flat}{\vdots}}} 
 \]
 By the induction hypothesis, we have proofs witnessing $\proves \phi$ and $\proves \psi$, and by $\land$-introduction:
 \[
 \infer{\phi\land \psi}{\deduce{\phi}{\vdots} & \deduce{\psi}{\vdots}}
 \]
 This completes the induction.

\end{proof}

Following the diagram, we have the completeness of IPL with respect to the B-eS:

\begin{proof} Theorem~\ref{thm:sandqvist} --- Completeness.  By definition, if $\supp \phi$, then $\supp_\base{N} \phi$. Hence, by Lemma~\ref{lem:emulation}, it follows that $T^\omega I_\bot, \base{N} \sat \phi^\flat$. By Lemma~\ref{lem:miller} $\base{N} \proves \phi^\flat$. Thus, by Lemma~\ref{lem:simulation}, $\proves \phi$, as required.
\end{proof}

In the following section, we discuss how reductive logic delivers the completeness proof above and the essential role played by both proofs and refutations.

\subsection{Negation-as-Failure} \label{sec:NAF}

A reduction in a proof system is constructed co-recursively by applying the rules of inference backwards. Even though each step corresponds to the application of a rule, the reduction can fail to be a proof as the computation arrives at an irreducible sequent that is not an instance of an axiom in the logic. For example, in \calculus{hHLP}, one may compute the following:
\[
\infer[\Uparrow]{
    \infer[\Uparrow]{\emptyset \seq \at{p} \to (\at{p} \lor \at{q})}{
        \at{p} \seq \at{p} \lor \at{q}}
}{\at{p} \seq \at{q}}
\]
This reduction fails to be a proof, despite every step being a valid inference, since the initial sequent is not an instance of $\rn{IN}$ or $\rn{ABSURD}$. In reductive logic, such failed attempts at constructing proofs are not meaningless. Pym and Ritter~\cite{pym2004reductive} have provided a semantics of the reductive logic of IPL in which such reductions are given meaning by using hypothetical rules; that is, the construction would succeed in the presence of the following rule:
\[
\infer{\at{q}}{\at{p}}
\]
 The categorical treatment of this semantics has them as \emph{indeterminates} in a polynomial category --- this adumbrates current work by Pym et al.~\cite{Pym2022catpts}, who have shown that the B-eS is entirely natural from the perspective of categorical logic. The use of such additional rules to give semantics to constructions that are not proofs directly corresponds to the use of atomic systems in the B-eS for IPL; for example, let $\base{A}$ be the atomic system containing the rule above, then the judgement $p \supp_\base{A} q$ obtains. Altogether, this suggests a close relationship between B-eS and reductive logic. We may review the meaning of absurdity ($\bot$) from this perspective.
 
 There is no introduction rule for $\bot$ in \calculus{NJ}. One may not construct a proof of absurdity without it already being, \emph{in some sense}, assumed; for example, $\phi, \phi \to \bot \proves \bot$ obtains because the context $\{\phi, \phi \to \bot\}$ is already, in some sense, absurd. We may use B-eS and LP to understand what that sense is. The judgement $\Gamma \proves \bot$ is equivalent to $ \proves \phi \to \bot$ for some formula $\phi$. Therefore, we may restrict attention to negations of this kind to understand the meaning of absurdity. 
 
 Using the work of Section~\ref{sec:compviaLP}, the judgement $\supp \neg \phi$ obtains iff  $T^\omega I_\bot, \base{N} \proves (\neg\phi)^\flat$. Unfolding the semantics, this is equivalent to $T^\omega I_\bot, \base{N} \cup \{\phi^\flat\} \proves \bot$. Thus, the sense in which $\phi$ is absurd is that its interpretation under $T^\omega I_\bot$ contains an absurdity; that is, $\phi$ is absurd iff $\bot \in T^\omega I_\bot(\phi)$. What does this tell us about the meaning of $\neg\phi$? We are passing through the following equivalence --- see $(\ast)$ in Section~\ref{sec:atomicprograms}:
\[
\supp_\base{B} \bot \qquad \text{ iff } \qquad \base{N}\cup \base{B} \proves \bot^\flat
\]
Recall that $\base{B}$ is finite in this setting. Hence, according to the LP perspective, what we mean by a base supporting absurdity is that it proves $\bot^\flat$. In this way, we introduce negation at the level of atomic propositions. That is, we may have have a base $\base{B}$ containing the following rules in which $\at{p}$ and $\bar{\at{p}}$ are both atoms:
\[
\infer{\,\bot^\flat\,}{\,\at{p} & \bar{\at{p}}\,}
\]
In this case, the inferential behaviour of $\at{p}$ and $\bar{\at{p}}$ is that they are contradictory propositions: together, they infer absurdity. Essentially, following the construction of $\base{N}$ in Section~\ref{sec:compviaLP}, we have $\at{p}= \phi^\flat$ and $\bar{\at{p}} = (\phi\to\bot)^\flat$, for some $\phi$.

This view of negation is in contrast to the semantics, originally proposed by Dummett~\cite{dummett1991logical}, in which the proof-theoretic meaning of absurdity is that all propositional atoms are proved; that is, the definition in which $\bot$ is understood by the following `virtually infinite' rule:
\[
\infer{\,\bot\,}{\,\at{p}_1 & ... & \at{p}_n\,}
\]
K\"urbis~\cite{kurbis2019proof} observes that this leaves something to be desired. 

The case in which a base proves every atomic proposition is degenerate because it corresponds to having every proof be valid. In the non-degenerate case, we may simply choose $\bot^\flat$ to be an atom that does not appear in $\base{N}\cup\base{B}$. Thus, the proof-theoretic meaning of $\bot$ is the failure to find a proof of $\bot$ while not working in a degenerate program. 

It follows, by the clauses of Figure~\ref{fig:support}, that the meaning of $\neg \phi$ is that there is no proof of $\phi$ while not working in a degenerate program,
\[
\begin{array}{lcl}
     \supp_\base{B} \neg \phi & \mbox{iff} & \phi \supp_\base{B} \bot  \\
     & \mbox{iff} & \mbox{$\supp_\base{C} \phi$ implies $\supp_\base{C} \bot$ (for $\base{C} \supseteq \base{B}$)} \\
     & \mbox{iff} & \mbox{$ \base{N} \cup \base{C} \not \proves \phi^\flat$ (unless $\base{B}$ degenerate)}
\end{array}
\]

Thus, B-eS supports negation-as-failure. In particular, since $\base{N}$ simulates $\calculus{NJ}$, the failure actually refers to failure to find a proof in the natural deduction system for IPL, even under extension by atomic rules, and not merely to the failure of hHLP to find a proof.

Piecha and Schroeder-Heister~\cite{Schroeder2016atomic,piecha2017definitional} have argued that there are two perspectives on atomic systems: the knowledge view and the definitional view. This becomes clear according to various ways in which a program may be regarded in LP. The negation-as-failure protocol makes use of the definitional perspective; its analogue in terms of knowledge is the \emph{closed-world assumption}. In this case, a knowledge base treats everything that is not known to be valid as invalid. There is significant literature about the closed-world assumption that may be useful for understanding P-tS and what it tells us about reasoning --- see, for example, Clark~\cite{clark1978negation}, Reiter~\cite{reiter1981closed}, and Kowalski~\cite{Kowalski1986,kowalskinotes}, and Harland~\cite{harland1991hereditary,harland1993success}.

\section{Conclusion} \label{sec:conclusion}

Proof-theoretic semantics is the paradigm of meaning based on proof (as opposed to truth). Essential to this approach is the use of atomic systems, which give meaning to atomic propositions. Base-extension semantics is a particular instance of proof-theoretic semantics that proceeds by an inductively defined judgement whose base case is given by provability in an atomic system. It may be regarded as capturing the declarative content of proof-theoretic semantics in the Dummett-Prawitz tradition --- see Gheorghiu and Pym~\cite{ptvtobes}. Sandqvist~\cite{Sandqvist2005inferentialist} has given a base-extension semantics for intuitionistic propositional logic. Completeness follows by constructing a special bespoke base in which the validity of a complex proposition simulates a natural deduction proof of that formula.

In the base-extension semantics, the meaning of the logical constants is derived from the rules of $\calculus{NJ}$, while the atomic systems give the meaning of atomic propositions. These atomic systems, which include Sandqvist's special bases that delivers completeness, all sit within the hereditary Harrop fragment of IPL. The significance of this is that an effective operational reading of definite formulae renders them meaning-conferring in a sense analogous to the use of atomic systems. Moreover, this operational account coheres with the independently conceived notion of derivability in an atomic system. Of course, that atomic systems and programs are intimately related has been studied before --- see Schroeder-Heister and Halln\"as~\cite{hallnas1990proof,hallnas1991proof}. 

Significantly, the operational reading of the definite formulae allows from a simple proof-theoretic model-theoretic semantics that captures the idea of \emph{unfolding} the inferential content of a set of definite clauses or an atomic system.  In this paper, we have used the operational account of definite formulae to prove the completeness of intuitionistic propositional logic with respect to its base-extension semantics. The aforementioned special base is interpreted as a program so that completeness follows immediately from the existing completeness result of the model-theoretic semantics of the logic programming language. Doing this reveals the subtle meaning of negation in proof-theoretic semantics.

Historically, the negation of a formula is understood as the denial of the formula itself. This is indeed the case in the model-theoretic semantics of IPL --- see Kripke~\cite{kripke1965semantical}. Using the connection to logic programming in this paper, we see that in base-extension semantics, negation is defined by the failure for there to be a proof. Thus, denial is conceptionally prior to negation. In short, base-extension semantics consider the space of reductions, which is larger than the space of proofs, including failed searches. 
As illustrated above, the connection between logic programming and base-extension semantics is quite intuitive and useful. More specifically, the $T$ operator delivering the semantics of logic programming corresponds to the application of a rule in a proof system; hence, the $T^\omega$ construction is fundamental to proof-theoretic semantics. Since logic programming has been studied for various logics (see, for example, the treatment of BI in Gheorghiu et al.~\cite{Samsonschrift}), this suggests the possibility for uniform approaches to setting up base-extension semantics for logics by studying their proof-search behaviours. In particular, work by Harland~\cite{harland1991hereditary,harland1993success} on handling negation in logic programming may be used to address the difficulties posed by the connective --- see K\"urbis~\cite{kurbis2019proof}. 

It remains to investigate further the connection between proof-theoretic semantics and reductive logic, in general, and base-extension semantics and logic programming, in particular.

\Acknowledgements{We are grateful to Edmund Robinson for suggesting the formula in Example~\ref{ex:counter} and to the reviewers of an earlier version of the paper for their helpful comments and feedback.  }

\bibliographystyle{BSLbibstyle}
\bibliography{biblio}

\end{document}